\documentclass[letterpaper,11pt]{amsart}
\usepackage{latexsym,array,delarray,amsthm,amssymb,epsfig}
\usepackage{caption}
\usepackage{subcaption}

\theoremstyle{plain}
\newtheorem{thm}{Theorem}[section]
\newtheorem{lemma}[thm]{Lemma}

\newtheorem*{thm*}{Theorem}
\newtheorem*{lemma*}{Lemma}
\newtheorem*{prop*}{Proposition}
\newtheorem*{cor*}{Corollary}
\newtheorem*{conj*}{Conjecture}

\theoremstyle{definition}
\newtheorem{defn}[thm]{Definition}

\theoremstyle{remark}

\newcommand{\zz}{\mathbb{Z}}
\newcommand{\nn}{\mathbb{N}}
\newcommand{\pp}{\mathbb{P}}

\newcommand{\rr}{\mathbb{R}}
\newcommand{\cc}{\mathbb{C}}

\newcommand{\ind}{\mbox{$\perp \kern-5.5pt \perp$}}

\begin{document}

\title{Tying Up Loose Strands: Defining Equations of the Strand Symmetric Model}

\author{Colby Long and Seth Sullivant}

\email{celong2@ncsu.edu}
\email{smsulli2@ncsu.edu } 
               
\address{Department of Mathematics, Box 8205, North Carolina State University, Raleigh, NC, 27695-8205, USA }

\begin{abstract}
The strand symmetric model is a phylogenetic model designed to reflect 
the symmetry inherent in the double-stranded structure of DNA. 
We show that the set 
of known phylogenetic invariants for the general strand symmetric model
 of the three leaf claw tree entirely defines the ideal. 
 This knowledge allows one to determine the vanishing ideal of the general strand symmetric model of any trivalent tree. 
 Our proof of the main result is computational.
 We use the fact that the Zariski closure of the 
  strand symmetric model is the secant variety
 of a toric variety to compute the  dimension of the variety. 
 We then show that the known equations generate a prime ideal 
 of the correct dimension
 using elimination theory. 
\end{abstract}
\maketitle

%
%

\section{Introduction}
\label{Introduction}

The strand symmetric model is a phylogenetic model 
 designed to reflect the symmetry inherent in the 
double-stranded structure of DNA. This symmetry naturally imposes restrictions on
 the transition probabilities assigned to each edge and imposing
only these restrictions gives the general strand symmetric model (SSM). 
The phylogenetic invariants of a model are algebraic relationships that 
must be satisfied by the probability distributions arising from the model. 
Their study was originally proposed as a method for reconstructing phylogenetic
trees \cite{Cavender, Lake1987}, but they have also been useful theoretical tools in proving
identifiability results (see e.g.~\cite{Allman2006}).
Results in \cite{Draisma2009} imply that to determine generators of the ideal 
of phylogenetic invariants for any trivalent tree, it suffices
to determine them for the 
 claw tree, $K_{1,3}$. 
 
Though the general strand symmetric model itself is not group-based,
 Casanellas and the second author \cite{Casanellas2005} showed that it is still amenable to the
 Fourier/Hadamard transform technique of \cite{Evans1993, Szekely}. 
In the Fourier coordinates, it becomes evident that the parameterization of the model for $K_{1,3}$ is a coordinate projection of the secant variety of the Segre embedding of $\pp^3 \times \pp^3 \times \pp^3$. 
From this observation, 
the same authors were able to find $32$ degree three and $18$ degree four invariants of the homogenous ideal for $K_{1,3}$ and to show that these invariants generate the ideal up to degree four. Whether or not these equations generate the entire ideal was heretofore unknown.
  
In this paper, we show that these  $50$ equations in fact generate the entire ideal of the SSM for $K_{1,3}$. First, we use the parameterization of the model after the matrix-valued Fourier transform and the tropical secant dimension
technique of Draisma \cite{Draisma} to determine the dimension of the variety of probability distributions arising from the model. 
Then, using Macaulay2 \cite{M2}, we show that the ideal generated by these fifty equations defines a variety of the same dimension.  Finally, with the aid of symbolic computation we  generate a decreasing sequence of elimination ideals demonstrating that the ideal in question is prime. Thus, the variety defined by these equations is irreducible, contains the parameterization, and is of the same dimension as the parameterization, from which the result follows.

%
%

\section{Phylogenetic Invariants of the SSM model}
\subsection{Preliminaries}
\label{Preliminaries}

 The general strand symmetric model on an $n$-leaf rooted tree $T$ is a phylogenetic model of 4-state character change. Since the
  SSM is specifically intended to model DNA evolution, we 
  associate to each node $v$ of the tree a random variable $X_v$ 
  with state space corresponding to the  DNA bases
  \{A,C,G,T\} . Associated to each edge is a $4\times4$ transition
   matrix with rows and columns indexed by the bases. The entry $\theta_{ij}$ encodes the probability of changing from character $i$ to $j$ along that edge. In the double helix structure of DNA it is always the case that  the bases A and T are paired together and likewise for C and G. So that our model reflects this strand symmetry, we let $\pi = (\pi_A, \pi_C, \pi_G, \pi_T)$ be the distribution of the bases at the root, and set $\pi_A = \pi_T$ and $\pi_C = \pi_G$. 
Additionally, since a character transition in one strand will induce a corresponding transition in the other, we insist
$$ \theta_{AA} = \theta_{TT}, \theta_{AC} = \theta_{TG}, \theta_{AT} = \theta_{TA}, \theta_{CA} = \theta_{GT}, \theta_{CC} = \theta_{GG}, \theta_{CG} = \theta_{GC}, \theta_{CT} = \theta_{GA}.$$

 The key observation from \cite{Casanellas2005} is that the SSM is a matrix-valued group-based model. Identify the character states of the random variables of a phylogenetic model with elements of $G \times \{ 0, \ldots, l \} $ where $G$ is a 
 finite abelian group. Then each character state is indexed by an element 
 $\left( \begin{smallmatrix} j \\ i \\ \end{smallmatrix} \right )$ 
 where $j \in G$ and $i \in \{ 0, \ldots, l \}$.  In these indices, the entries of the transition matrix along edge $E$ are written $E^{j_1j_2}_{i_1 i_2}$ and the probability that the root is in state
$\left( \begin{smallmatrix} j \\ i \\ \end{smallmatrix} \right )$ 
is equal to $R_i^j$.
\begin{defn} 
A phylogenetic model is a \emph{matrix-valued group-based model} if for each edge, the matrix transition probablities satisfy 
 $$E^{j_1j_2}_{i_1 i_2} =  E^{k_1k_2}_{i_1 i_2}$$ 
 whenever $j_1 - j_2$ = $k_1 - k_2$ and the root distribution probabilities satisfy $R_i^j = R_i^k$.
\end{defn}
\noindent Let $G = \zz_2$ and $l = 1$, then the following identifications 
make manifest the matrix-valued group-based structure of the SSM:
$A = \left ( \begin{smallmatrix} 0 \\ 0 \\ \end{smallmatrix} \right )$,
$C = \left ( \begin{smallmatrix} 0 \\ 1 \\ \end{smallmatrix} \right )$,
$G = \left ( \begin{smallmatrix} 1 \\ 0 \\ \end{smallmatrix} \right )$,
$T = \left ( \begin{smallmatrix} 1 \\ 1 \\ \end{smallmatrix} \right )$.

The tree parameter of an algebraic model determines a polynomial map sending each choice of stochastic parameters into the probability space indexed by $n$-tuples of the characters. Thus, for the SSM of a tree $T$, we have the following map
$$ \phi_T : S_T \rightarrow \Delta^{4^n - 1} .$$
If we do not impose the stochastic conditions on the parameters then $\overline{ \rm{im}( \phi_T)}$, where the closure is taken in the Zariski topology, is a variety. In Section 16.1 of \cite{Casanellas2005}, the authors detail the group-valued Fourier 
transform and show how it can be used to obtain a simple parameterization for the closure of the cone over the SSM for $T = K_{1,3}$, denoted $CV(T)$.
Letting $q_{ijk}^{mno}$ be the transformed coordinates of the image space, we have 

$$
\psi : q_{ijk}^{mno} = d^{mm}_{0i}e^{nn}_{0j}f^{oo}_{0k} + d^{mm}_{1i}e^{nn}_{1j}f^{oo}_{1k}
 $$
 if $m + n + o \equiv 0 $ in $\zz_2$, and 
 $q_{ijk}^{mno} = 0 $  otherwise. 
Now to determine the defining equations for the SSM for $K_{1,3}$, it is enough to determine the defining equations for $\overline{ \rm{im}( \psi_T)} = CV(T)$. Let $I$ be the ideal generated by the fifty equations found in \cite{Casanellas2005}, the rest of the paper will be concerned with proving the following theorem.

\begin{thm} \label{main}
The vanishing ideal of the strand symmetric model for the graph $K_{1,3}$
is minimally generated by $32$ cubics and $18$ quartics.  The ideal
has dimension $20$, degree $9024$, and Hilbert series
$$ \frac{ 1 + 12t + 78t^{2}  + 332t^{3}  + 984t^{4}  + 1908t^{5}  + 2394t^{7}  + 1908t^{8}  + 984t^{9}  + 332t^{10}  + 78t^{11}   + 12t^{12}   + t^{13}}{(1-t)^{20}}.
$$
\end{thm} 

Note that the Hilbert series suggests that the ideal is Gorenstein though
we have not been able to prove this.

\subsection{Dimension}
\label{Dimension}

A toric variety is a variety that is parametrized by monomials. 
Let $C \subset CV(T)$ be the toric variety parameterized  in each coordinate only by the monomial containing variables with 
zero in the first entry of the subscript.  
With this definition, $CV(T)$ is the secant variety $C*C$ and so  we can use existing techniques from \cite{Draisma} for computing the dimensions of secant varieties. 

The theorem from \cite{Draisma} which we wish to apply is conveniently formulated for our purposes by Theorem 15 from \cite{Allman}. We associate to each monomial $x_1^{u_1} x_2^{u_2} \ldots x_n^{u_n}$ in the parameterization of a toric variety an integer vector $u$ and let $A$ be the set of these integer vectors. 
Let $H = \{x \in \rr^d : c^Tx = e \}$ be a hyperplane in $\rr^d$ that splits $\rr^d$ into two components which we will label $H^+  = \{x \in \rr^d : c^Tx > e \}$ and $H^- = \{x \in \rr^d : c^Tx < e \}$. 

In our case, the matrix $A$ is a $12 \times 32$ matrix of rank $10$, with each column containing exactly threes $1$'s and nine $0$'s. 
If we let $\{e^{0}_{0}, e^{0}_{1}, e^{1}_{0}, e^{1}_{1} \} $ denote the standard basis in $\rr^{2 \times 2}$ then
the thirty-two columns of $A$ are
$$
\{ e^{m}_{i} \oplus  e^{n}_{j} \oplus e^{o}_{k}  \in \rr^{12} :  m + n + o \equiv 0 \mbox{ in } \zz_{2} \}.
$$

\begin{thm}\cite[Theorem 15]{Allman}
\label{hyperplane}
Let $V_A$ be a projective toric variety with corresponding set of exponent vectors $A \subset \nn^d$.  
Let $H$ be a hyperplane not intersecting $A$. Let $A^+ = A \cap H^+$ and $A^- = A \cap H^-$. Then $\text{dim} ( V_A * V_A ) \geq \text{rank}(A^+) + \text{rank}(A^-) - 1$. 
\end{thm}

\begin{lemma}
\label{dim}
 dim$(CV(T)) = dim(V(I)) = 20.$
\end{lemma}
\begin{proof}
Regard $C$ as a projective variety so that $C = V_A$ from Theorem \ref{hyperplane}. The hyperplane defined by the vector $c = (1,1,0,0,1,1,0,0,1,1,0,0)$ and $e = \frac{3}{2}$ gives 
$|A^+| = |A^-| = 16$ and  rank$(A^+)$ = rank$(A^-) =10 $. Therefore, by Theorem \ref{hyperplane}, as a projective variety dim($C*C) \geq 19$ and as an affine cone dim$(CV(T)) \geq 20$.
Using Macaulay2 we determine that $\text{dim}(V(I)) = 20$, and since $CV(T) \subseteq V(I)$, we must have  dim$(CV(T))= 20$. 
\end{proof}

\subsection{Primality}
\label{Primality}

In this section we outline our approach for determining if the ideal $I$ is prime. 

There are algorithms for determining whether or not an ideal is prime implemented in many computer algebra systems.  However, these algorithms do not terminate for many of the large ideals confronted in practice, including
the ideal $I$ generated by the cubics and quartics contained in $I(CV(K_{1,3}))$. 
We use the following result from \cite{Garcia2005}
which in certain cases allows one to determine the primality of an ideal by determining the primality of an ideal in fewer variables.

\begin{lemma}
\label{primelemma}
 \cite[Proposition 23]{Garcia2005} Let $k$ be a field and  $J \subset k[x_1, \ldots, x_n]$ be an ideal containing a polynomial $f = gx_1 + h$ with $g,h$ not involving $x_1$ and $g$ a non-zero divisor modulo $J$. Let $J_1 \cap k[x_2, \ldots, x_n]$ be the elimination ideal. Then $J$ is prime if and only if $J_1$ is prime.
\end{lemma}

Proposition 23 of \cite{Garcia2005} was stated without proof, so we include
a proof of the result for completeness.

\begin{proof} ($\Rightarrow$) It is true in general that the elimination ideal of a prime ideal is prime. Suppose $J$ is prime and let $a,b \in  k[x_1, \ldots, x_n] \setminus J_1 $ such that $ab \in J_1$. Since $J_1 \subset J$, it must be that either $a$ or $b$ is in $J \setminus J_1$, otherwise it would contradict that $J$ is prime. Therefore, either $a$ or $b$ is in $k[x_1, \ldots, x_n] \setminus k[x_2, \ldots, x_n]$ and so $ab$ must have some term that involves $x_1$, which implies $ab \not \in J_1$, a contradiction. 

\medskip

\noindent ($\Leftarrow$)
Suppose $J_1$ is prime but that $J$ is not. Then there must 
exist $ a,b \in  k[x_1, \ldots, x_n] \setminus J $ with $ab \in J \setminus J_1$.
 Choose $a$ and $b$ so 
that $ab$ has minimal $x_1$-degree among all such pairs. Let $d$ be the $x_1$-degree of $a$ and $d'$ the $x_1$-degree of $b$. 
Since $ab \in J \setminus J_1$, $d + d' \geq 1$, and so  
without loss of generality we can assume $d \geq 1$. Write
$$a = h_0 + h_1x_1 + h_2x_1^2 + \ldots + h_{d}x_1^{d},$$ where each $h_i \in 
k[x_2, \ldots, x_n]$ and $h_{d} \not = 0$. Then since $f \in J$ and $g$ is not a zero divisor mod $J$, 
$a' :=  (ga - h_{d}x_1^{d - 1}f )$ is not in $J$ and has $x_1$-degree strictly less than $d$. It follows that 
$a'b $ has $x_1$-degree strictly less than that of $ab$. Finally, since $ab$ and $f$ are in $J$, 
$a'b = gab - h_{d}x_1^{d - 1}f b$ is 
 in $J $, contradicting the minimality of the $x_1$-degree of $ab.$
 \end{proof}
 
 \begin{lemma}
 \label{lisprime}
The ideal $I$ generated by the $32$ cubics and $18$ quartics of the general strand symmetric model for $K_{1,3}$ is prime.
\end{lemma}

\begin{proof} 

The proof is obtained by repeated application of Lemma \ref{primelemma}. The computations we describe can be found at 

\begin{center}
{\tt http://www4.ncsu.edu/$\sim$smsulli2/Pubs/LooseStrandsWebsite/SSM\_Supplement.html}
\end{center}
in the Macaulay2 file 
{\tt SSM\_Supplement} where the symbols 0,1,2, and 3 are substituted for $\left ( \begin{smallmatrix} 1 \\ 1 \\ \end{smallmatrix} \right )$,
$\left ( \begin{smallmatrix} 1 \\ 0 \\ \end{smallmatrix} \right )$,
$\left ( \begin{smallmatrix} 0 \\ 1 \\ \end{smallmatrix} \right )$, and
$\left ( \begin{smallmatrix} 0 \\ 0 \\ \end{smallmatrix} \right )$.

First, we let $I_0= I$. Beginning with $k = 1$, we find a polynomial $f_k = g_kx_k + h_k \in I_{k-1}$, verify that $g_k$ is not a zero-divisor mod  $I_{k-1}$, and then eliminate the corresponding variable to obtain the ideal $I_{k}$. In this way we generate a decreasing chain of elimination ideals 
$$
I = I_0 \supset I_1 \supset I_2 \ldots \supset I_{10} .
$$
Using the {\tt isPrime} function in Macaulay2, we show that $I_{10}$, and hence every ideal in the sequence, is prime.

\end{proof}

While this is the general outline of our approach, it is actually computationally 
easier to show that none of the $g_k$ that we encounter is a zero-divisor mod the respective elimination ideal first. 
Identify the new indices $0,1,2,$ and $3$ with the set of standard basis vectors $ \{ e_1, e_2, e_3, e_4 \}$ and define a multi-grading where the weight of  $q_{ijk}$ is $e_{i + 1} \oplus e_{j + 1} \oplus e_{ k + 1}$.
Let $q_{\alpha}q_{\beta} - q_{\gamma}q_{\delta}$ be a nontrivial binomial that is homogenous  with respect to this grading. As it so happens, we are always able to choose $f_k = g_k x + h_k$ so that $g_k$ is either a binomial of this form or a product of binomials of this form. There are two elementary observations that will be useful: 
\begin{enumerate}
\item $g = l_1l_2$ is a zero-divisor mod $J$ if and only if at least one of $l_1$ and $l_2$ is.
\item $g$ is not a zero-divisor mod any elimination ideal of $J$ if it is not a zero-divisor  mod $J$. 
\end{enumerate}
Thus, to show that none of the $g_k$ is a zero-divisor mod $I_{k - 1}$ it is enough to show that none of the homogenous binomials is a  zero-divisor mod $I$.

The symmetry of $I$ enables us to establish this by considering only a small subset of the homogenous binomials. There is a group action of $S_4 \times S_4 \times S_4\rtimes S_3$  on $Sec^2(Seg(\mathbb{P}^3\times\mathbb{P}^3\times\mathbb{P}^3 ))$, that comes from performing the rank-preserving column and transposition operations. Hence, the same group acts on $I(Sec^2(Seg(\mathbb{P}^3\times\mathbb{P}^3\times\mathbb{P}^3 ))))$, where 
 column operations correspond to interchanging the symbols in the indices of the variables and transposition operations correspond to permuting the order of the indices of each variable.  Let $G $ be the subgroup of elements of $ S_4 \times S_4 \times S_4\rtimes S_3$ satisfying $g \cdot q_{ijk}^{mno} = q_{i'j'k'}^{m'n'o'} $ with $m + n + o \equiv m' + n' + o'$ in $\zz_2$ for each of the 64 variables. Since 
  $$I(CV(T)) = I(Sec^2(Seg(\mathbb{P}^3\times\mathbb{P}^3\times\mathbb{P}^3 )))) \cap \cc[q_{ijk}^{mno} : m + n + o = 0], $$ 
$G$ acts on $I(CV(T))$, and since the generators of $I$ generate $I(CV(T))$ up to degree four, $G$ acts on $I$ as well.
Let $H$ be the subgroup of $G$ generated by elements that correspond to interchanging symbols in the indices. For example, $h = ((01),(01)(23), (01)) \in H$ interchanges $0 \leftrightarrow 1$ in the first index, $0\leftrightarrow1$ and $2 \leftrightarrow 3$ in the second, and $0 \leftrightarrow 1$ in the third so that $h\cdot (q_{021}q_{113} - q_{013}q_{121}) = (q_{130}q_{003} - q_{103}q_{030})$.  Then 
\begin{align*}
 H = \langle  & ((01) , id , id) , ( id, (01) , id) , ( id , id , (01)),
                     ((23) , id , id) , ( id, (23) , id) , \\  &  ( id , id , (23)),   
                     ((0213) , (0213) , id) , ( (0213), id , (0213))
        \rangle
\end{align*} is a 256-element normal subgroup and $G \cong H \rtimes S_3$. The set of homogeneous binomials partitions into three orbits under the action of $G$. In the file {\tt SSM\_Supplement} we show that none of the homogeneous binomials is a zero-divisor by showing that one representative of each of the orbits under the group action is not a zero-divisor.

 Having shown that $I$ is prime, we are able to give a short proof of  Theorem \ref{main}.

\begin{proof}[Proof of  Theorem \ref{main}]
The containment $I \subset I(CV(T))$ implies that $CV(T) \subset V(I)$. By Lemma \ref{lisprime}, $I$ is prime and so $V(I)$ is an irreducible variety. By Lemma \ref{dim}, $CV(T)$ is an irreducible variety contained in an irreducible variety of the same dimension, so $CV(T) = V(I)$ and $I = I(CV(T))$. Knowing explicit generators of the vanishing ideal of the strand symmetric model for the graph $K_{1,3}$, the claims about the rank, degree, and Hilbert Series of the ideal are easily verified by the Macaulay2 code in {\tt SSM\_Supplement}.
\end{proof}

 \bibliography{references_1}

\begin{thebibliography}{10}

\bibitem{Allman}
E.S. Allman, S.~Petrovic, J.A. Rhodes, and S.~Sullivant.
\newblock Identifiability of 2-tree mixtures for group-based models.
\newblock {\em IEEE/ACM Trans Comput Biol Bioinformatics}, 8(3):710--722, 2011.

\bibitem{Allman2006}
E.S. Allman and J.A. Rhodes.
\newblock The identifiability of tree topology for phylogenetic models,
  including covarion and mixture models.
\newblock {\em J. Comp. Biol.}, 13(5):1101--1113, 2006.

\bibitem{Casanellas2005}
Marta Casanellas and Seth Sullivant.
\newblock {\em Algebraic Statistics for Computational Biology}, chapter~16.
\newblock Cambridge University Press, Cambridge, United Kingdom, 2005.

\bibitem{Cavender}
J.A. Cavender and J.~Felsenstein.
\newblock Invariants of phylogenies in a simple case with discrete states.
\newblock {\em J. of Class.}, 4:57--71, 1987.

\bibitem{Draisma}
J.~Draisma.
\newblock A tropical approach to secant dimensions.
\newblock {\em J. Pure Appl. Algebra}, 212(2):349--363, 2008.

\bibitem{Draisma2009}
Jan Draisma and Jochen Kuttler.
\newblock On the ideals of equivariant tree models.
\newblock {\em Math. Ann.}, 344(3):619--644, 2009.

\bibitem{Evans1993}
S.N. Evans and T.P. Speed.
\newblock Invariants of some probability models used in phylogenetic inference.
\newblock {\em Ann. Statist}, 21(1):355--377, 1993.

\bibitem{Garcia2005}
Luis~David Garcia, Michael Stillman, and Bernd Sturmfels.
\newblock Algebraic geometry of bayesian networks.
\newblock {\em Journal of Symbolic Computation}, 39(3-4):331--355, March-April
  2005.

\bibitem{M2}
D.R. Grayson and M.E. Stillman.
\newblock Macaulay2, a software system for research in algebraic geoemetry.
\newblock Available at http://www.math.uiuc.edu/Macaulay2/, 2002.

\bibitem{Lake1987}
J.~A. Lake.
\newblock A rate-independent technique for analysis of nucleaic acid sequences:
  evolutionary parsimony.
\newblock {\em Molecular Biology and Evolution}, 4:167--191, 1987.

\bibitem{Szekely}
L.~Sz\'{e}kely, P.L. Erd\"{o}s, M.A. Steel, and D.~Penny.
\newblock A fourier inversion formula for evolutionary trees.
\newblock {\em Applied Mathematics Letters}, 6(2):13--17, 1993.

\end{thebibliography}
\bibliographystyle{plain}
\end{document}